\documentclass[journal,transmag]{IEEEtran}
\usepackage[comma,numbers,square,sort&compress]{natbib}

\usepackage{amsmath}
\usepackage{amssymb}
\usepackage{float}
\usepackage{amsthm}
\usepackage{graphicx}
\usepackage{epstopdf}
\usepackage{color}
\usepackage{verbatim}
\usepackage{tikz,times}
\usepackage{algorithm}
\usepackage{algorithmic}
\newtheorem{definition}{\textbf{Definition}}
\newtheorem{remark}{\textbf{Remark}}
\newtheorem{theorem}{\textbf{Theorem}}
\newtheorem{lemma}{\textbf{Lemma}}
\newtheorem{assumption}{\textbf{Assumption}}

\allowdisplaybreaks[4]
\begin{document}

	\title{Consistent distributed state estimation with global observability over sensor network}
	
	\author{Xingkang~He,
		~Wenchao~Xue,
  Haitao Fang
		
		\thanks{Xingkang He, Wenchao Xue and Haitao Fang  are with the Key Laboratory of System and Control,  Academy of Mathematics and Systems Science,
			Chinese Academy of Sciences, Beijing, China. (xkhe@amss.ac.cn, wenchaoxue@amss.ac.cn, htfang@iss.ac.cn)}
		
	}
	
	
	\maketitle

	\begin{abstract}
	This paper studies the distributed state estimation problem for a class of discrete time-varying systems over sensor networks. Firstly, it is shown that a networked Kalman filter with optimal gain parameter    is actually a centralized filter, since it requires each sensor to have global information which is usually forbidden in large networks. Then, a sub-optimal distributed Kalman filter (DKF)  is proposed by employing the covariance intersection (CI) fusion  strategy. It is proven that the proposed DKF is of consistency, that is, the upper bound of error covariance matrix can be provided by the filter in real time. The consistency also enables the design of adaptive CI weights for better filter precision. Furthermore, the boundedness of covariance matrix and the convergence of the proposed filter are proven based on the strong connectivity of directed network topology and the global observability which permits the sub-system with local sensor's measurements to be unobservable. Meanwhile, to keep the covariance of the estimation error bounded, the proposed DKF does not require the system matrix to be nonsingular at each moment, which seems to be a necessary condition in the main DKF designs under global observability. Finally, simulation results of two examples show the effectiveness of the algorithm in the considered scenarios. 
	\end{abstract}

	\section{Introduction}
	Wireless sensor networks (WSNs) usually consist of intelligent sensing devices located at different geographical positions. Since multiple sensors can collaboratively carry out the task by information communication via the wireless channels, WSNs have been widely applied in environmental monitoring \cite{Cao2008Development}, collaborative information processing \cite{Kumar2012Computational}, data collection \cite{Solis2007In}, distributed signal estimation \cite{Schizas2008Consensus}, and etc.
	In the past decades, state estimation problems of WSNs have drawn more and more attention of researchers. Two approaches are usually considered in existing work. The first one is centralized filtering, i.e., a data center is set to collect  measurements from all sensors at each sampling moment. The centralized Kalman filter (CKF) can be directly designed such that the minimum variance state estimator is achieved for
	linear systems with Gaussian noises. However, the centralized frame is fragile since it could be easily influenced by link failure, time delay, package loss and so on. The second approach, on the contrary, utilizes distributed strategy, in which no central sensor exists. The implementation of this strategy simply depends on information exchange between neighbors \cite{Das2015Distributed, khan2011coor, Khan2014Collaborative,yang2017stochastic,speranzon2008distributed,farina2016partition,boem2015distributed,Sun2016Dynamic}. Compared with the centralized approach, the distributed frame  has stronger ability in robustness and parallel processing.

	Information communication between sensors plays an important role in the design of distributed filtering. 
	Generally, communication rate between neighbors could be faster than the rate of measurement sensing. Fast information exchange between neighbors supports the consensus strategy which can achieve the agreement of information variables (e.g. measurements \cite{Das2015Distributed}) of sensors. Actually, \cite{olfati2007distributed, khan2008distributing, carli2008distributed, Cat2010Diffusion} have shown some remarkable results on the convergence and the consensus of local filters with the consensus strategy. 
	However, faster communication rate probably needs larger capability of computation and transmission to conduct the consensus before the updates of filters.
	In the single-time scale, the neighbor communication and measurement sensing share the same rate, which can not only 
	reduce communication burden but also result in computation cost linearly matching with sensor number over the network \cite{Khan2014Collaborative, Matei2012Consensus, zhou2013distributed, liu2015event}.
	Additionally, the DKF algorithm with faster communication rate can be designed by combining the filter with single-time scale and the consensus process.
	Hence, this paper considers distributed state estimation algorithms in the single-time scale.

	Parameter design  of algorithms is one of the most essential parts in the study of distributed state estimation problems.
	{\color{black}
		In \cite{He16con}, it is shown that a networked Kalman filter with optimal gain parameter   is actually a centralized filter since the calculation of time-varying gain parameter is dependent on information of non-neighbors. Then a modified sub-optimal distributed filter under undirected graph is proposed.}
	Distributed filters with constant filtering gains are well studied in \cite{Khan2014Collaborative, khan2011coor, Khan2010On}, which evaluate the relationship between the instability of system and the boundedness of estimation error.
	In \cite{Das2015Distributed}, measurement consensus based DKF is presented and design methods of the consensus weights as well as the filtering gains are rigorously studied.
	In \cite{Cat2010Diffusion}, a general diffusion DKF  based on time-invariant weights is proposed and  performance of the distributed algorithm is analyzed in detail.
	To achieve better estimation precision, time-varying parameters are considered in \cite{speranzon2008distributed}, which provides a distributed minimum variance estimator for a scalar time-varying signal.

	In  \cite{boem2015distributed}, a distributed prediction method for dynamic systems is proposed to minimize bias and variance. The  method can effectively compute  time-varying weights of the distributed algorithm.
	A scalable partition-based distributed Kalman filter is investigated in \cite{farina2016partition} to deal with  coupling terms and uncertainty among  sub-systems. Furthermore, stability of this algorithm is guaranteed through designing proper parameters. 
	Nevertheless, the work mentioned above have not considered the distributed filter problem with global observability condition which allows the sub-system with local sensor's measurements to be unobservable. 

	Research of distributed filter for time-varying systems based on global observability is an important but difficult problem.
	Since sensors of WSNs are sparsely located in different positions, the observability condition assumed for the sub-system with respect to one sensor is much stronger than that assumed for the overall system based on  global network. 
	However, the work mentioned above pay little attention to boundedness  analysis of covariance matrix and convergence analysis of the algorithm under global observability. Regarding
	time-invariant systems, conditions on global observability are usually determined by the system matrix, the network topology and the global observation matrix which collects model information of all sensors \cite{Khan2014Collaborative,khan2011coor,Khan2010On}. This means that  distributed filters with constant filtering gain can be designed to guarantee  stability of the algorithm. However, most of the methods fail for  time-varying systems. \cite{Battistelli2014Kullback, Battistelli2015Consensus} give some pioneer work on building consensus DKF algorithms under the global observability for time-invariant systems. Nevertheless, they require the assumption that the system matrix is nonsingular, which seems to be severe for time-varying systems at every moment.
	In this paper, we aim to develop a scalable and totally distributed algorithm for a class of discrete linear time-varying systems in the WSNs. The main contributions are summarized as follows.
	\begin{enumerate}
		\item The proposed consistent distributed Kalman filter (CDKF) guarantees the error covariance matrix can be upper bounded by a parameter matrix, which is timely calculated 	by each sensor using local information. This property is quite of importance since it supports an effective error evaluation principle in real time.
		\item A set of adaptive weights based on CI fusion is determined through a Semi-definite Programming (SDP) convex optimization method. It is proven that the proposed adaptive CI weights ensure lower error covariance bound than that with constant CI weights which are mainly used in existing work \cite{Battistelli2014Kullback, Battistelli2015Consensus, Battistelli2016stability}. Therefore, adaptive CI weights can achieve  improvement of estimation performance.
		\item Global observability instead of local observability is assumed for the  system over networks. This allows the sub-system with local sensor's measurements to be unobservable. {\color{black}Additionally, the assumption of system matrix being nonsingular at each moment  is loosened \cite{Battistelli2014Kullback, Battistelli2015Consensus, Battistelli2016stability,He16con,Wang2017On}}. Since the nonsingularity of system matrix at each moment is difficult to be satisfied for time-varying systems, the proposed filter can greatly enlarge application range of the distributed state estimation algorithms.
	\end{enumerate}
	
	The remainder of this paper is organized as follows. Section 2 presents some necessary preliminaries and notations of this paper. Section 3 is on  problem formulation and distributed filtering algorithms. Section 4 considers performance of the proposed algorithm. Section 5 is on  simulation studies. The conclusion of this paper is given in Section 6.
	
	\section{Preliminaries and Notations}
	
	Let $\mathcal{G=(V,E,A)}$ be a directed graph, which consists of  the set of nodes $\mathcal{V}=\{1,2,\cdots,N\}$, the set of edges $\mathcal{E}\subseteq \mathcal{V}\times \mathcal{V}$ and the weighted adjacent matrix $\mathcal{A}=[a_{i,j}]$. In the weighted adjacent matrix $\mathcal{A}$, all elements are nonnegative, row stochastic and the diagonal elements are all positive, i.e., $a_{i,i}> 0,a_{i,j}\geq 0,\sum_{j\in \mathcal{V}}a_{i,j}=1$. If $a_{i,j}>0,j\neq i$, then there is an edge $(i,j)\in \mathcal{E}$, which means Node $i$ can directly receive the information of Node $j$. In this situation, Node $j$ is called the neighbor of Node $i$. All neighbors of Node $i$ including itself can be represented by the set $\{j\in\mathcal{V}|(i,j)\in \mathcal{E}\}\bigcup\{i\}\triangleq \mathcal{N}_{i}$, whose size is denoted as $|\mathcal{N}_{i}|$. $\mathcal{G}$ is called strongly connected if for any pair nodes $(i_{1},i_{l})$, there exists a directed path from $i_{1}$ to $i_{l}$ consisting of edges $(i_{1},i_{2}),(i_{2},i_{3}),\cdots,(i_{l-1},i_{l})$. According to \cite{horn2012matrix} and \cite{varga2009matrix}, the following lemma is obtained.
	\begin{lemma}\label{lem_primitive}
		If the directed graph $\mathcal{G=(V,E,A)}$ is strongly connected with $\mathcal{V}=\{1,2,\cdots,N\}$, then all elements of  $\mathcal{A}^{s},s\geq N-1,$ are positive.
	\end{lemma}

	Throughout this paper, the notations used are fairly standard.  The superscript ``T" represents  transpose. The notation $A\geq B$ (or $A>B$), where $A$ and $B$ are both symmetric matrices, means that $A-B$ is a positive semidefinite (or positive definite) matrix. $I_{n}$ stands for the identity matrix with $n$ rows and $n$ columns. $E\{x\}$ denotes the mathematical expectation of the stochastic variable $x$, and  $blockcol\{\cdot\}$ means the block elements are arranged in columns. $blockdiag\{\cdot\}$ and $diag\{\cdot\}$ represent the diagonalizations of block elements and scalar elements, respectively. $tr(P)$ is the trace of matrix $P$. The notation $\otimes$ stands for  tensor product. The integer set from $a$ to $b$ is denoted as $[a:b]$.
	\section{Problem Formulation and Distributed Filtering Algorithms}
	Consider the following time-varying stochastic system
	\begin{equation}\label{system_all}
	\begin{cases}
	x_{k+1}=A_{k}x_{k}+\omega_{k},\quad k=0,1,2,...,\\
	y_{k,i}=H_{k,i}x_{k}+v_{k,i},\quad i=1,2,\cdots,N,
	\end{cases}
	\end{equation}
	where
	$x_{k}\in\mathbb{R}^{n}$ is the state at the $k$th moment, $A_{k}\in\mathbb{R}^{n\times n}$ is the known system matrix, $\omega_{k}\in\mathbb{R}^{n}$ is the process noise with  covariance matrix $Q_{k}\in\mathbb{R}^{n\times n}$,
	$y_{k,i}\in\mathbb{R}^{m}$ is the measurement vector obtained via Sensor $i$,  $H_{k,i}\in\mathbb{R}^{m\times n}$ is the observation matrix of Sensor $i$ and $v_{k,i}$ is the observation noise with covariance matrix $R_{k,i}\in\mathbb{R}^{m\times m}$. 
	$N$ is the number of sensors over the network.

	{\color{black}
		\begin{definition}\label{def_1}
			The $i$th sub-system of the overall system (\ref{system_all}) is defined as the system with respect to $(A_{k},H_{k,i})$.
		\end{definition}
		
		In this paper, the following assumptions are needed.
		\begin{assumption}\label{ass_noise}
			The sequences $\{\omega_{k}\}_{k=0}^{\infty}$ and $\{v_{k,i}\}_{k=0}^{\infty}$ are zero-mean, Gaussian,  white and uncorrelated. Also, $R_{k,i}$ is positive definite, $\forall k\geq0$. There exist two constant positive definite matrices $\bar Q_1$ and $\bar Q_2$ such that $\bar Q_1\leq Q_{k}\leq\bar Q_2, \forall k\geq 0$. The initial state $x_{0}$ is generated by a zero-mean white Gaussian process independent of $\{\omega_{k}\}_{k=0}^{\infty}$ and $\{v_{k,i}\}_{k=0}^{\infty}$, subject to $E\{x_{0}x_{0}^T\}=P_{0}$.
		\end{assumption}

		\begin{assumption}\label{ass_observable}
			The system	(\ref{system_all}) is uniformly completely observable, i.e., there exist a positive integer $\bar N$ and positive constants $\alpha,\beta$ such that
			\begin{equation*}
			0<\alpha I_{n}\leq \sum_{j=k}^{k+\bar N}\Phi^T_{j,k}H_{j}^TR_{j}^{-1}H_{j}\Phi_{j,k}\leq \beta I_{n},
			\end{equation*}
			for any $k\geq 0$,
			where
			\begin{align*}
			\begin{cases}
			\Phi_{k,k}=I_{n},\Phi_{k+1,k}=A_{k},\Phi_{j,k}=\Phi_{j,j-1}\cdots \Phi_{k+1,k},\\
			H_{k}= blockcol\{H_{k,1},H_{k,2},\cdots,H_{k,N}\},\\
			R_{k}=blockdiag\{R_{k,1},R_{k,2},\cdots,R_{k,N}\}.
			\end{cases}
			\end{align*}
		\end{assumption}
		\begin{assumption}\label{ass_topology}
			The topology of the network $\mathcal{G=(V,E,A)}$ is a fixed directed graph and it is strongly connected.
		\end{assumption}
		\begin{assumption}\label{ass_A_bound}
			There exists a positive scalar $\beta_{1}$, such that
			\begin{equation*}
			\lambda_{max}(A_{k}A_{k}^T)\leq\beta_{1}, \forall k\geq0.
			\end{equation*}
		\end{assumption}
		
		\begin{assumption}\label{ass_A}
			There exist a sequence set $\mathcal{K}=\{k_{l},l\geq1\}$, an integer $L\geq N+\bar N$ and a scalar $\beta_{2}>0$, such that
			\begin{equation*}
			\begin{cases}
			\sup_{l\geq1} (k_{l+1}-k_{l})<\infty,\\
			\inf_{l\geq1} (k_{l+1}-k_{l})>0,\\
			\lambda_{min}(A_{k_{l}+s}A_{k_{l}+s}^T)\geq\beta_{2},\forall k_{l}\in \mathcal{K},s=0,\cdots,L-1.
			\end{cases}
			\end{equation*}
		\end{assumption}

		\begin{remark}
			Assumption $\ref{ass_observable}$ is a basic global observability condition which does not require any sub-system with local sensor's measurements to be observable.  
			Assumption \ref{ass_topology} is quite general for the direct topology graph of the network, since strong connectivity is the basic condition for the implementation of distributed algorithms which rely on information spread over the networks.
			Assumption \ref{ass_A} does not require $A_{k}$ to be nonsingular at each moment \cite{Battistelli2015Consensus,Battistelli2016stability,He16con,Wang2017On}.
		\end{remark}
	}
	
	%
	%
	
	In this paper, we consider the following general distributed filtering structure for  Sensor $i$ , which mainly consists of three parts:
	\begin{equation*}
	\begin{cases}
	\bar x_{k,i}=A_{k-1}\hat x_{k-1,i},\\
	\phi_{k,i}=\bar x_{k,i}+K_{k,i}(y_{k,i}-H_{k,i}\bar x_{k,i}),\\
	\hat x_{k,i}=\sum_{j\in \mathcal{N}_{i}}W_{k,i,j}\phi_{k,j},\\
	\qquad\text{   s.t. }\sum_{j\in \mathcal{N}_{i}}W_{k,i,j}=I_{n},W_{k,i,j}=0, \text{if } j\notin \mathcal{N}_{i},
	\end{cases}
	\end{equation*}
	where $\bar x_{k,i}$,  $\phi_{k,i}$ and $\hat x_{k,i}$ are the state prediction,  state update and  state estimate of Sensor $i$ at the $k$th moment, respectively.
	$K_{k,i}$ is the filtering gain matrix and $W_{k,i,j}$ is the local fusion matrix.  Additionally, the condition $\sum_{j\in \mathcal{N}_{i}}W_{k,i,j}=I_{n}$ is to guarantee the unbiasedness of the estimates.


	The  design of optimal filtering gain matrix $K_{k,i}^*$ can be achieved through
	\begin{equation*}
	K_{k,i}^*=arg \min_{K_{k,i}} tr(P_{k,i}),
	\end{equation*}
	where $P_{k,i}=E\{(\hat x_{k,i}-x_{k})(\hat x_{k,i}-x_{k})^T\}$.
	Then one can obtain the networked Kalman filter with optimal gain parameter in Table \ref{ODKF1}, where $K_{k,i}$ stands for $K_{k,i}^*$ hereafter for convenience \cite{He16con}. In this algorithm, the error covariance matrices  are derived with
	the forms $ \bar P_{k,i,j}=E\{(\bar x_{k,i}-x_{k})(\bar x_{k,j}-x_{k})^T\}$, $\tilde P_{k,i,j}=E\{(\phi_{k,i}-x_{k})(\phi_{k,j}-x_{k})^T\}$ and $ P_{k,i,j}=E\{(\hat x_{k,i}-x_{k})(\hat x_{k,j}-x_{k})^T\}$. However, since the calculations of $(\bar P_{k,i,j}, \tilde P_{k,i,j}, P_{k,i,j})$ need the global information on $\{K_{k,j},H_{k,j},j\in\mathcal{V}\}$. {\color{black}The algorithm 1 in Table \ref{ODKF1} is actually a centralized filter, which is almost impossible to be conducted in a scalable manner
		for a large network. }

	Since the optimal design of $W_{k,i,j}$ depends on the covariance matrices which rely on global information \cite{He16con}, we will discuss the sub-optimal design for $W_{k,i,j}$ simply with the local information in the following text. {\color{black}Generally, for the design of local fusion weights $W_{k,i,j}$, the traditional methods assume $W_{k,i,j}=\alpha_{i,j}I_{n}$, where $\alpha_{i,j}$ are positive scalars satisfying the required conditions (\cite{Cat2010Diffusion,Wenling2015Diffusion}).} In this paper, $W_{k,i,j}$ are considered as time-varying matrix weights obtained by the CI strategy \cite{Julier1997A}. Hence, we propose a sub-optimal scalable algorithm named as consistent distributed Kalman filter in Table \ref{ODKF2}, which corresponds to the communication topology illustrated in Fig. \ref{topology1}. {\color{black} In the communication process,  only the pair ($\phi_{k,j}$, $\tilde P_{k,j}$ ) is transfered from neighbors. }
	

	\begin{table}
		\caption{Networked Kalman Filter with Optimal Gain \cite{He16con}:}
		\label{ODKF1}
		\begin{tabular}{l}  		
			\hline\hline
			\textbf{Prediction:}\\
			$\bar x_{k,i}=A_{k-1}\hat x_{k-1,i},$\\  
			$\bar P_{k,i}=A_{k-1}P_{k-1,i}A_{k-1}^T+Q_{k-1},$\\         
			$\bar P_{k,i,j}=A_{k-1}P_{k-1,i,j}A_{k-1}^T+Q_{k-1},$\\
			\textbf{Measurement Update:}\\
			$\phi_{k,i}=\bar x_{k,i}+K_{k,i}(y_{k,i}-H_{k,i}\bar x_{k,i})$,\\        
			$K_{k,i}=\bar P_{k,i}H_{k,i}^T(H_{k,i}\bar P_{k,i}H_{k,i}^T+R_{k,i})^{-1}$,\\
			$\tilde P_{k,i}=(I-K_{k,i}H_{k,i})\bar P_{k,i}$,\\
			$\tilde P_{k,j,s}=(I-K_{k,j}H_{k,j})\bar P_{k,j,s}(I-K_{k,s}H_{k,s})^T,j\neq s$,\\
			\textbf{Local Fusion:}\\
			$\hat x_{k,i}=\sum_{j\in \mathcal{N}_{i}}W_{k,i,j}\phi_{k,j}$,	\\
			$P_{k,i}=\sum_{j\in \mathcal{N}_{i}}\sum_{s\in \mathcal{N}_{i}}W_{k,i,j}\tilde P_{k,j,s}W_{k,i,s}^T$,\\
			$P_{k,i,l}=\sum_{j\in \mathcal{N}_{i}}\sum_{s\in \mathcal{N}_{l}}W_{k,i,j}\tilde P_{k,j,s}W_{k,l,s}^T$.\\
			\hline
		\end{tabular}
	\end{table}
	
	\begin{table}
		\caption{Consistent Distributed Kalman Filter:}
		\label{ODKF2}
		\begin{tabular}{l}  
			\hline\hline
			\textbf{Prediction:}\\
			$\bar x_{k,i}=A_{k-1}\hat x_{k-1,i},$\\  
			$\bar P_{k,i}=A_{k-1}P_{k-1,i}A_{k-1}^T+Q_{k-1},$\\         
			\textbf{Measurement Update:}\\
			$\phi_{k,i}=\bar x_{k,i}+K_{k,i}(y_{k,i}-H_{k,i}\bar x_{k,i})$,\\        
			$K_{k,i}=\bar P_{k,i}H_{k,i}^T(H_{k,i}\bar P_{k,i}H_{k,i}^T+R_{k,i})^{-1}$,\\
			$\tilde P_{k,i}=(I-K_{k,i}H_{k,i})\bar P_{k,i}$,\\
			\textbf{Local Fusion:} Receiving ($\phi_{k,j}$, $\tilde P_{k,j}$ ) from neighbors $j\in \mathcal{N}_{i}$\\
			$\hat x_{k,i}=P_{k,i}\sum_{j\in \mathcal{N}_{i}}w_{k,i,j}\tilde P_{k,j}^{-1}\phi_{k,j}$,\\	
			$P_{k,i}=(\sum_{j\in \mathcal{N}_{i}}w_{k,i,j}\tilde P_{k,j}^{-1})^{-1}$,	\\
			Design $w_{k,i,j}(\geq 0) $, such that $\sum_{j\in \mathcal{N}_{i}}w_{k,i,j}=1$.\\	
			\textbf{Initialization:}\\
			$\hat x_{0,i}=0, P_{0,i}\geq P_{0}$. \\
			\hline
		\end{tabular}
	\end{table}


	\begin{figure}[htp]
		\begin{center}
			\begin{tikzpicture}[scale=0.8, transform shape,line width=2pt]
			\node  [circle,shade, fill=gray!30] (a) at (0, 0) {Sensor 1};
			\node  [circle,shade, fill=gray!30] (b) at +(0: 1.5*3) {Sensor 2};
			\node  [circle,shade, fill=gray!30] (c) at +(45: 1.5*2) {Sensor 3};
			\node  [circle,shade, fill=gray!30] (d) at +(180: 2) {$\cdots$};	
			\node  [circle,shade, fill=gray!30] (e) at +(27: 1.5*3.1) {$\cdots$};	
			\node  [circle,shade, fill=gray!30] (f) at +(0: 1.5*4.35) {$\cdots$};
			\node  [circle,shade, fill=gray!30] (g) at +(90: 1.5*1.4) {$\cdots$};			
			\foreach \from/\to in {a/b, b/c, c/a}
			\draw [black!30,->] (\from) -- (\to);
			\draw [black!30,->] (c) -- (e);
			\draw [black!30,->] (d) -- (a);
			\draw [black!30,->] (g) -- (c);
			\draw [black!30,->] (f) -- (b);
			\path (0:1.5) node(text1)  [below]{$(\phi_{k,1} \text{,} \tilde P_{k,1})$};
			\path (17:1.5*2.35) node(text2)  [right]{($\phi_{k,2} \text{,} \tilde P_{k,2})$};
			\path (19:1.8*1.7) node(text3)  [left]{$(\phi_{k,3} \text{,} \tilde P_{k,3})$};
			\end{tikzpicture}
		\end{center}
		\caption{An illustration of the communication topology for consistent distributed Kalman filter}\label{topology1}
	\end{figure}
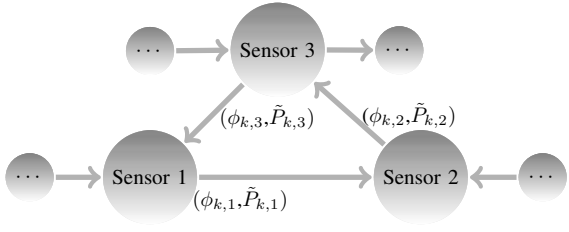
	
	Some remarks on the proposed consistent distributed Kalman filter in Table \ref{ODKF2} are given as follows. Firstly,
	the matrix $P_{k,i}$ in the algorithm  may not stand for the error covariance matrix of Sensor $i$. In the subsequent parts, we will show the relationship between $P_{k,i}$ and the error covariance matrix.  Secondly, the time-varying CI weights $\{w_{k,i,j}\}$ are considered as constant CI weights in \cite{Battistelli2014Kullback, Battistelli2015Consensus, Battistelli2016stability}. This paper will show an adaptive design method with respect to  $\{w_{k,i,j}\}$ through a convex optimization algorithm. Thirdly,
	the proposed algorithm can also be equipped with the consensus (multiple times of local fusion)  with certain steps similar to \cite{Battistelli2014Kullback}.  {\color{black}Fourthly, the computation complexity of the CDKF in Table \ref{ODKF2} for Sensor $i$ is $O(m^3+n^3|\mathcal{N}_{i}|)$ if $\{w_{k,i,j}\}$ are set to be constant, such as $w_{k,i,j}=a_{i,j}$. If we turn to obtain 
		the optimized weights by certain optimization algorithms, 
		the computational complexity of the total algorithm should include the complexity  of the specific optimization method. }

	In the next section, we will give the performance analysis of the proposed consistent distributed Kalman filter.
	\section{Performance Analysis}
	\subsection{ Error Evaluation and Consistency}
	Firstly, the following theorem shows the state estimation error's probability distribution of each sensor.
	\begin{theorem}\label{thm_distri}
		Consider the system (\ref{system_all}) with the CDKF in Table \ref{ODKF2}, then under Assumption \ref{ass_noise} the state estimation error of each sensor is zero-mean and Gaussian,  i.e., the  following equation holds
		\begin{equation}\label{lem_unbias}
		\hat x_{k,i}-x_{k}=e_{k,i}\sim \mathcal{N}(0,E\{e_{k,i}e_{k,i}^T\}), \forall i\in \mathcal{V}, k\geq 0,
		\end{equation}
		where $\mathcal{N}(0,U)$ is the Gaussian distribution with mean $0$ and covariance matrix $U$.
	\end{theorem}
	\begin{proof}
		See Appendix \ref{pf_dis}.
	\end{proof}
	The Gaussianity and unbiasedness of estimation error in Theorem \ref{thm_distri} provide an effective evaluation method for the system state, if we can obtain the estimation error covariance matrix $E\{e_{k,i}e_{k,i}^T\}$.
	In the Kalman filter, the error covariance can be represented by the parameter $P_{k}$. However, in the distributed Kalman filters \cite{Das2015Distributed,olfati2007distributed,Cat2010Diffusion}, the relationship between $P_{k,i}$ and error covariance matrix is uncertain.
	For the sake of evaluating the estimation error of CDKF, their relationship will be analyzed from the aspect of consistency defined as follows.
	\begin{definition}(\cite{Julier1997A})
		Suppose $x_{k}$ is a random vector. Let $\hat x_{k}$ and $P_{k}$ be the estimate of $x_{k}$ and the estimate of the corresponding error covariance matrix. Then the pair ($\hat x_{k},P_{k}$) is said to be consistent (or of consistency) at the $k$th moment if
		\begin{equation*}
		E\{(\hat x_{k}-x_{k})(\hat x_{k}- x_{k})^T\}\leq P_{k}.
		\end{equation*}
	\end{definition}
	The following theorem shows the consistency of CDKF, which directly depicts the relationship between the estimation error covariance matrix  $E\{e_{k,i}e_{k,i}^T\}$ and the parameter matrix $P_{k,i}$.
	\begin{theorem}\label{theorem_consistent}
		Considering the system (\ref{system_all}),  {\color{black}under Assumption \ref{ass_noise}, }the pair ($\hat x_{k,i},P_{k,i}$) of the CDKF in Table \ref{ODKF2} is consistent, i.e.,
		\begin{equation}\label{P_consistent}
		E\{(\hat x_{k,i}-x_{k})(\hat x_{k,i}-x_{k})^T\}\leq P_{k,i}, \forall i\in \mathcal{V}, k\geq 0.
		\end{equation}
	\end{theorem}
	\begin{proof}
		Here we utilize a inductive method to finish the proof of this theorem.	
		Firstly, under the initial condition, due to $x_{0}\sim \mathcal{N}(0,P_{0})$, there is $E\{(\hat x_{0,i}-x_{0})(\hat x_{0,i}-x_{0})^T\} \leq P_{0,i}$.
		It is supposed that, at the $(k-1)$th moment, $E\{(\hat x_{k-1,i}-x_{k-1})(\hat x_{k-1,i}-x_{k-1})^T\}=E\{e_{k-1,i}e_{k-1,i}^T\} \leq P_{k-1,i}.$
		The equation (\ref{e_bar1}) provides the prediction error at the $k$th moment with the form
		$\bar e_{k,i}=A_{k-1}e_{k-1,i}-\omega_{k-1}.$
		Due to $E\{e_{k-1,i}\omega_{k-1}^T\}=0$, it is immediate to  see that
		$E\{\bar e_{k,i}\bar e_{k,i}^T\}=A_{k-1}E\{e_{k-1,i}e_{k-1,i}^T\}A_{k-1}^T+Q_{k-1}.$
		Thus, $E\{\bar e_{k,i}\bar e_{k,i}^T\}\leq A_{k-1}P_{k-1,i}A_{k-1}^T+Q_{k-1}=\bar P_{k,i}.$
		In the update process, according to  (\ref{e_k1}), there is
		$\tilde e_{k,i}=(I-K_{k,i}H_{k,i})\bar e_{k,i}+K_{k,i}v_{k,i}.$
		Because of $E\{\bar e_{k,i}v_{k,i}^T\}=0$, we can obtain $E\{\tilde e_{k,i}\tilde e_{k,i}^T\}\leq \tilde P_{k,i}.$
		Notice that $e_{k,i}=\hat x_{k,i}-x_{k}=P_{k,i}\sum_{j\in \mathcal{N}_{i}}w_{k,i,j}\tilde P_{k,j}^{-1}\tilde e_{k,j},$
		then we get
		\begin{equation}\label{e_k010}
		\begin{split}
		E\{e_{k,i}e_{k,i}^T\}=&P_{k,i}E\{\Delta_{k,i}\} P_{k,i}.
		\end{split}
		\end{equation}
		where \\
		$\Delta_{k,i}=\big(\sum_{j\in \mathcal{N}_{i}}w_{k,i,j}\tilde P_{k,j}^{-1}\tilde e_{k,j}\big)\big(\sum_{j\in \mathcal{N}_{i}}w_{k,i,j}\tilde P_{k,j}^{-1}\tilde e_{k,j}\big)^T$.
		
		According to the consistent estimate of covariance intersection \cite{Niehsen2002Information}, we find that
		\begin{equation}\label{e_k0101}
		E\{\Delta_{k,i}\} \leq \sum_{j\in \mathcal{N}_{i}}w_{k,i,j}\tilde P_{k,j}^{-1}E\{\tilde e_{k,j}\tilde e_{k,j}^T\}\tilde P_{k,j}^{-1}
		\leq P_{k,i}^{-1}.
		\end{equation}
		Hence, the combination of  (\ref{e_k010}) and  (\ref{e_k0101}) leads to (\ref{P_consistent}). \textbf{Q.E.D.}
	\end{proof}
	\begin{remark}
		Theorem \ref{theorem_consistent} essentially states that the estimation error covariance matrix can be upper bounded by the parameter $P_{k,i}$ provided by the filter. {\color{black}With this property, one  can not only  evaluate the estimation error in real time under the error distribution illustrated in Theorem \ref{thm_distri}, but also  judge the boundedness of covariance matrix through $P_{k,i}$.}

	\end{remark}
	\subsection{Design of Adaptive CI Weights}
	Since the matrix $P_{k,i}$ is the upper bound of the covariance matrix of the unavailable estimation error, we seek to compress $P_{k,i}$ so as to lower the estimation error.
	Fortunately, the proper design of $w_{k,i,j}$ is helpful to achieve the compression on $P_{k,i}$.
	%
	%
	%
	{\color{black}
		Since $P_{k,i}=(\sum_{j\in \mathcal{N}_{i}}w_{k,i,j}\tilde P_{k,j}^{-1})^{-1}$, 
		we aim to obtain a smaller $P_{k,i}$ than $P_{k,i}$ calculated with constant weights $a_{i,j}$, i.e., 
		\begin{align}\label{optim_0}
		\Delta_{k,i}=\sum_{j\in \mathcal{N}_{i}}(w_{k,i,j}-a_{i,j})\tilde P_{k,j}^{-1}>0.
		\end{align}	
		Under the condition $\Delta_{k,i}>0$, we consider the 
		following optimization problem:
		\begin{align}\label{optim_1}
		\{w_{k,i,j},j\in \mathcal{N}_{i}\}=arg\min_{w_{k,i,j}} tr(\Delta_{k,i}^{-1}),
		\end{align}
		where $\sum_{j\in \mathcal{N}_{i}} w_{k,i,j}=1,0\leq  w_{k,i,j}\leq 1,j\in \mathcal{N}_{i}.$
		
		The reason that we choose (\ref{optim_1}) as the objective function lies in the fact that when $tr(\Delta_{k,i}^{-1})$ decreases, $tr(\Delta_{k,i})$ increases. Under the condition $\Delta_{k,i}>0$, we can achieve a compression of $P_{k,i}$ at each channel compared with constant weights $a_{i,j}$.
		To solve the problem (\ref{optim_1}), it is equivalently transfered to another form given in Lemma \ref{lemma_1}.
		
		
		\begin{lemma}\label{Schur_Complement}(Schur Complement \cite{Seber2007A})
			The following linear matrix inequality (LMI):
			\begin{equation*}
			\left(
			\begin{array}{cc}
			Q(x) & S(x) \\
			S(x)^T & R(x) \\
			\end{array}
			\right)>0,
			\end{equation*}
			where $Q(x)=Q(x)^T$ and $R(x)=R(x)^T$, is equivalent to the following condition:
			\begin{align*}
			Q(x)>0,\quad  R(x)-S(x)^TQ(x)^{-1}S(x)>0.
			\end{align*}
		\end{lemma}
		
		\begin{lemma}\label{lemma_1}
			The problem (\ref{optim_1}) is equivalent to the following convex optimization problem
			\begin{align}\label{objective_weight}
			\{w_{k,i,j},m_{k,i_l}\}=arg\min_{w_{k,i,j},m_{k,i_l}} tr(M_{k,i}),
			\end{align}
			subject to
			\begin{align}\label{eq_constraint1}
			&\left(
			\begin{array}{cc}
			\Delta_{k,i} & I_{n} \\
			I_{n} & M_{k,i} \\
			\end{array}
			\right)> 0,
			\end{align}
			where $M_{k,i}=diag\{m_{k,i_1},m_{k,i_2},\cdots,m_{k,i_n}\}> 0$, $\sum_{j\in \mathcal{N}_{i}} w_{k,i,j}=1,0\leq  w_{k,i,j}\leq 1,j\in \mathcal{N}_{i}.$
		\end{lemma}
		\begin{proof}
			See Appendix \ref{pf_lemma1}
		\end{proof}
		

		%
		%
		%
		
		\begin{lemma}\label{lemma_2}
			The problem (\ref{objective_weight}) is a convex optimization problem, which can be solved through  the algorithm of Table \ref{SDP}\footnote{The subscript $j$ of the non-zero parameters ($ w_{k,i,j}, a_{i,j}, P_{k,j}$) is set from $1$ to $|\mathcal{N}_{i}|$ without loss of generality.} based on the popular SDP algorithm \cite{boyd2004convex}.	
		\end{lemma}
		\begin{proof}
			See Appendix \ref{pf_lemma2}.
		\end{proof}	
	}
	
	\begin{table}[htp]
		\caption{SDP algorithm for the solution of (\ref{objective_weight})}
		\label{SDP}
		\begin{tabular}{l}  
			\hline\hline
			\textbf{Input:    } $ a_{i,j},\tilde P_{k,j},j=1,\cdots,|\mathcal{N}_{i}|,$ \\
			\textbf{Output:} $w_{k,i,j},$ obtained by solving the following problem \\
			1) \textbf{SDP optimization:} \\
			$\quad \quad\quad\quad\quad\quad \quad \quad z^*=\arg\min c^Tz$, \\
			subject to 		$A_{1}z=b_{1}$, $A_{2}z\geq b_{2}$,\\ $F_{0}+z_{1}F_{1}+\cdots+z_{|\mathcal{N}_{i}|+n}F_{|\mathcal{N}_{i}|+n}>0$,\\
			where 
			$b_{1}\in \mathbb{R},$ $z,c $, $A_1\in \mathbb{R}^{|\mathcal{N}_{i}|+n,1}$, $A_2\in \mathbb{R}^{2|\mathcal{N}_{i}|+n,|\mathcal{N}_{i}|+n}
			$,\\ 
			$b_2\in \mathbb{R}^{2|\mathcal{N}_{i}|+n,1}$,
			$ F_s\in \mathbb{R}^{2n,2n},s=0,\cdots,|\mathcal{N}_{i}|+n$,\\
			with the following forms\\
			$z=\left(	w_{k,i,1}-a_{i,1} \cdots w_{k,i,|\mathcal{N}_{i}|}-a_{i,|\mathcal{N}_{i}|}\quad m_{k,i_n} \cdots m_{k,i_n}\right)^T$,\\
			$c=\left(
			\begin{array}{cccccc}
			0 & \cdots & 0 & 1 & \cdots & 1 \\
			\end{array}
			\right)^T$,$A_{1}=\left(
			\begin{array}{cccccc}
			1 & \cdots & 1 & 0 & \cdots & 0   \\
			\end{array}
			\right),$\\		
			$ b_2=\left(
			\begin{array}{cccccc}
			-a_{i,1} & a_{i,1}-1 &\cdots & -a_{i,|N_{i}|} & 1-a_{i,|\mathcal{N}_{i}|} & 0^{1\times n} \\
			\end{array}
			\right)^T,$\\
			$b_1=0$, $A_2=\left(
			\begin{array}{cc}
			I^{|\mathcal{N}_{i}|\times|\mathcal{N}_{i}|}\otimes (1\quad -1)^T
			& 0 \\
			0 & I^{n\times n} \\
			\end{array}
			\right),$\\
			$F_s=\left(
			\begin{array}{cc}
			\tilde P_{k,s}^{-1} & 0 \\
			0 & 0^{n\times n} \\
			\end{array}
			\right),s\leq|\mathcal{N}_{i}|,$ $F_0=\left(
			\begin{array}{cc}
			0 & I^{n\times n} \\
			I^{n\times n} & 0 \\
			\end{array}
			\right),$\\		
			$F_s=diag\{\underbrace{0,\cdots,0}_{n+s-|\mathcal{N}_{i}|-1},1,0,\cdots,0\},
			|\mathcal{N}_{i}|< s\leq |\mathcal{N}_{i}|+n$.\\
			2) \textbf{Adaptive CI weights:} \\
			$\quad\quad w_{k,i,j}=z^*_j+a_{i,j},j=1,\cdots,|\mathcal{N}_{i}|.$\\\hline
			
		\end{tabular}
	\end{table}

	\begin{remark}
		The result of the SDP algorithm in Table \ref{SDP} may not be feasible due to the constraint of LMI (\ref{eq_constraint1}). 
		If the feasibility of the SDP algorithm is not satisfied, $\{w_{k,i,j}\}$ will keep the setting $\{w_{k,i,j}=a_{i,j},j\in\mathcal{N}_i\}$. Thus, the proposed algorithm always works no matter the SDP algorithm is feasible or not.
	\end{remark}
	{\color{black}
		\begin{remark}
			Generally, the state dimension is  not  high in practical applications, thus the SDP optimization of Table \ref{SDP} can be well handled  with the existing optimization algorithms, such as interior point methods, first-order methods and Bundle methods.  
			The total computational complexity of the proposed algorithm consists of the filtering method and the chosen optimization method. 
		\end{remark}
	}
	By utilizing the algorithm of Table \ref{SDP} to solve the optimization problem  (\ref{objective_weight}), one can obtain a set of adaptive CI weights at each moment, which gives rise to effective compression on the error covariance bound.
	Regarding the comparison of the error covariance bound between the adaptive CI weights and the constant CI weights, the following theorem gives a direct conclusion.
	\begin{theorem}\label{thm_compare_P}
		Considering the system (\ref{system_all}) with the CDKF in Table \ref{ODKF2}, under Assumption \ref{ass_noise} and the same initial conditions,
		there is
		\begin{equation*}
		P_{k,i|w}\leq P_{k,i|a},\quad\forall i\in\mathcal{V},\forall k\geq 0,
		\end{equation*}
		where the parameter matrix $P_{k,i|w}$ and $P_{k,i|a}$ correspond to $w_{k,i,j}$ and  $a_{i,j}$, respectively.
	\end{theorem}
	\begin{proof}
		See Appendix \ref{pf_compare}.
	\end{proof}
	\subsection{Boundedness and Convergence of CDKF}
	
	{\color{black}Due to the consistency of CDKF in Theorem \ref{theorem_consistent}, the boundedness of $P_{k,i}$ implies the boundedness of covariance matrix. Thus, we draw the following boundedness conclusion on the proposed CDKF.}
	\begin{theorem}\label{thm_consistent}
		{\color{black}	Under Assumptions \ref{ass_noise}--\ref{ass_A}, }there exists a positive definite matrix $\hat P$, such that
		\begin{equation}\label{thm_compare}
		P_{k,i}\leq \hat P<\infty, \quad   \forall i\in \mathcal{V},\forall k\geq 0,
		\end{equation}
		where $P_{k,i}$ is the parameter matrix of the CDKF in Table \ref{ODKF2}.
	\end{theorem}
	\begin{proof}
		According to Theorem \ref{thm_compare_P} and $P_{k,i}=P_{k,i|w}$, we only need to prove $ P_{k,i|a}$ can be uniformly upper bounded.
		Under Assumption \ref{ass_A}, one can pick out a subsequence set $\{k_{l_{m}},m\geq 1\}$ from the sequence set $\mathcal{K}=\{k_{l},l\geq 1\}$ such that $L\leq k_{l_{m+1}}-k_{l_{m}}$. Due to $\sup_{l\geq1} (k_{l+1}-k_{l})<\infty$, there exists a sufficiently large integer $\bar L$, such that  $L\leq k_{l_{m+1}}-k_{l_{m}}\leq\bar L$. Without loss of generality, we suppose the set $\{k_{l}\}$ has this property, i.e., $L\leq k_{l+1}-k_{l}\leq\bar L$, $\forall l\geq 1$ .
		To prove the boundedness of $ P_{k,i|a}$, we divide the sequence set $\{k_{l},l\geq 1\}$ into two  bounded and non-overlapping set :  $\{k_{l}+L,l\geq 1\}$ and $\bigcup_{l\geq 1} [k_{l}+L+1:k_{l+1}+L-1]$.

		Step 1: $k=k_{l}+L$, $l\geq1$.	
		
		At the $(k_{l}+L)$th moment, substituting $w_{k,i,j}=a_{i,j}$ into the CDKF, there is $P_{k_{l}+L,i|a}^{-1}=\sum_{j\in \mathcal{N}_{i}} a_{i,j}\tilde{P}_{k_{l}+L,j|a}^{-1}.$
		According to (\ref{equ_short}) and Assumptions \ref{ass_A_bound} and \ref{ass_A}, we can obtain
		\begin{align}\label{proof_stability13}
		P_{k_{l}+L,i|a}^{-1}&=\sum_{j\in \mathcal{N}_{i}} a_{i,j}(\bar P_{k_{l}+L,j|a}^{-1}+H_{k_{l}+L,j}^TR_{k_{l}+L,j}^{-1}H_{k_{l}+L,j})\nonumber\\
		&\geq  \eta\sum_{j\in \mathcal{N}_{i}}a_{i,j}A_{k_{l}+L-1}^{-T}P_{k_{l}+L-1,j|a}^{-1} A_{k_{l}+L-1}^{-1}\nonumber\\
		&\quad+\sum_{j\in \mathcal{N}_{i}} a_{i,j}H_{k_{l}+L,j}^TR_{k_{l}+L,j}^{-1}H_{k_{l}+L,j},
		\end{align}
		where $0< \eta<1$, and the last inequality is derived similarly to Lemma 1 in \cite{Battistelli2014Kullback} by noting the lower boundedness of  $A_{k_{l}+L-1}A_{k_{l}+L-1}^T$ and upper boundedness of $Q_{k}$.

		By recursively applying  (\ref{proof_stability13}) for $L$ times, there is
		\begin{equation}\label{proof_stability3}
		\begin{split}
		P_{k_{l}+L,i|a}^{-1}\geq& \eta^{L}\Phi_{k_{l}+L,k_{l}}^{-T}(\sum_{j\in \mathcal{V}} a_{i,j}^{L}P_{k_{l},j|a}^{-1})\Phi_{k_{l}+L,k_{l}}^{-1}\\
		&+\sum_{s=1}^{L}\eta^{s-1}\sum_{j\in \mathcal{V}} a_{i,j}^{s}S_{k_{l}+L-s+1,j},
		\end{split}
		\end{equation}
		where $\Phi_{k,j}$ is the state transition matrix defined in Assumption \ref{ass_observable} and 
		\begin{align*}
		\begin{cases}
		S_{k-s+1,j}=\Phi_{k,k-s+1}^{-T}\bar S_{k-s+1,j}
		\Phi_{k,k-s+1}^{-1}\\
		\bar S_{k-s+1,j}=H_{k-s+1,j}^T R_{k-s+1,j}^{-1}H_{k-s+1,j}.
		\end{cases}
		\end{align*}
		According to Assumption \ref{ass_topology} and Lemma \ref{lem_primitive}, there is $a^{s}_{i,j}>0,s\geq N$.
		Since the first part on the right side of  (\ref{proof_stability3}) is positive definite,  we consider the second part
		denoted as $\breve{P}_{k,i|a}^{-1}$. Then
		\begin{align}\label{proof_stability4}
		&\breve{P}_{k,i|a}^{-1}\nonumber\\
		=&
		\sum_{s=1}^{L}\eta^{s-1}\sum_{j\in \mathcal{V}} a_{i,j}^{s}S_{k_{l}+L-s+1,j}\nonumber\\
		\geq&\sum_{s=N}^{L}\eta^{s-1}\sum_{j\in \mathcal{V}} a_{i,j}^{s}S_{k_{l}+L-s+1,j}\nonumber\\
		\geq& a_{min}\sum_{s=N}^{L}\eta^{s-1}S_{k_{l}+L-s+1,j}\nonumber\\
		\geq& a_{min}\eta^{L-1}\sum_{b=k_{l}+1}^{k_{l}+1+L-N}\Phi_{k_{l}+L,b}^{-T}
		\sum_{j\in \mathcal{V}}H_{b,j}^T R_{b,j}^{-1}H_{b,j}\Phi_{k_{l}+L,b}^{-1}\nonumber\\
		\geq &a_{min}\eta^{L-1}\sum_{b=k_{l}+1}^{k_{l}+1+L-N}\Phi_{k_{l}+L,b}^{-T}
		H_{b}^T R_{b}^{-1}H_{b}\Phi_{k_{l}+L,b}^{-1},
		\end{align}
		where $a_{min}=arg\min_{i,j\in \mathcal{V}}{a_{i,j}^{s}>0,s\in [N:L]}$, $H_{k}$ and $R_{k}$ are defined in Assumption $\ref{ass_observable}$.

		Since the system (\ref{system_all}) is uniformly completely observable, there is
		\begin{align}\label{eq_obser}
		\sum_{j=k_{l}+1}^{k_{l}+1+\bar N}\Phi^T_{j,k_{l}+1}H_{j}^TR_{j}^{-1}H_{j}\Phi_{j,k_{l}+1}&=G_{k_{l}+1}^T\hat R_{k_{l}+1}^{-1}G_{k_{l}+1}\nonumber\\
		&\geq \alpha I_{n},\alpha>0,
		\end{align}
		where
		\begin{equation*}
		\begin{split}
		&G_{k}=blockcol\{H_{k},H_{k+1}\Phi_{k+1,k},\cdots,H_{k+\bar N}\Phi_{k+\bar N,k}\},\\
		&\hat R_{k}=blockdiag\{R_{k}, R_{k+1},\cdots,R_{k+\bar N}\}.
		\end{split}
		\end{equation*}
		Due to the nonsingularity of $A_{k},k\in[k_{l}:k_{l}+L-1],$ and the relationship $N+\bar N\leq L$, the matrix $F_{k_{l}+1}$ can be well defined as $F_{k_{l}+1}=\Phi_{k_{l}+1+\bar N,k_{l}+1}^{-1}.$
		Under Assumption \ref{ass_A}, for $k\in[k_{l}:k_{l}+L-1]$, there exists a positive real $\kappa$, such that $F_{k_{l}+1}^TF_{k_{l}+1}>\kappa I_{n}$.
		Thus, considering (\ref{eq_obser}), one can obtain that
		\begin{align}\label{eq_obser2}
		&\sum_{j=k_{l}+1}^{k_{l}+1+\bar N}\Phi^{-T}_{k_{l}+1+\bar N,j}H_{j}^TR_{j}^{-1}H_{j}\Phi_{k_{l}+1+\bar N,j}^{-1}\\
		=&F_{k_{l}+1}^TG_{k_{l}+1}^T\hat R_{k_{l}+1}^{-1}G_{k_{l}+1}F_{k_{l}+1}\geq \kappa\alpha I_{n},\alpha>0,\kappa>0.\nonumber
		\end{align}
		According to (\ref{eq_obser2}) and $L\geq N+\bar N$, $\breve{P}_{k,i|a}^{-1}$ in  (\ref{proof_stability4}) is lower bounded by a constant positive definite matrix. Thus, $P_{k_{l}+L,i|a}$ is upper bounded by a constant positive definite matrix $P$, i.e., $P_{k_{l}+L,i|a}\leq P$.
		
		\textbf{Step 2:} $k\in[k_{l}+L+1:k_{l+1}+L-1]$, $l\geq1$.
		
		Under Assumption \ref{ass_A_bound}, there is $A_{k}A_{k}^T\leq \beta_{1}I_{n}$. Due to $Q_{k}\leq\bar Q_2<\infty$, for $k\in[k_{l}+L+1:k_{l+1}+L-1]$, there is
		\begin{equation*}
		\begin{split}
		&P_{k,i|a}\\
		&\leq \big(\lambda_{max}(P_{k-1,i|a})\beta_{1}+\lambda_{max}(\bar Q_2)\big)I_{n}\\
		&\leq \big(\lambda_{max}(P_{k-2,i|a})\beta_{1}^2+\lambda_{max}(\bar Q_2)\beta_{1}+\lambda_{max}(\bar Q_2)\big)I_{n}\\
		&\qquad\qquad\vdots\\
		&\leq \bigg(\lambda_{max}(P)\beta_{1}^{k-k_{l}-L}+\sum_{j=0}^{k-k_{l}-L-1}\lambda_{max}(\bar Q_2)\beta_{1}^{j}\bigg)I_{n}\\
		&\leq \bigg(\lambda_{max}(P)\beta_{1}^{k_{l+1}-k_{l}-1}+\sum_{j=0}^{k_{l+1}-k_{l}-2}\lambda_{max}(\bar Q_2)\beta_{1}^{j}\bigg)I_{n}\\
		&\leq \bigg(\lambda_{max}(P)\beta_*^{\bar L-1}+\sum_{j=0}^{\bar L-2}\lambda_{max}(\bar Q_2)\beta_{1}^{j}\bigg)I_{n}\triangleq P^{mid},
		\end{split}
		\end{equation*}
		where $\beta_*=\max{(\beta_{1},1)}$.
		Since $k_{1}$ is finite, for $k\in[0:k_{1}+L-1]$, there exists a constant matrix $P^0$, such that $P_{k,i|a}\leq P^0$.
		Given $P$, $P^{mid}$ and $P^0$, according to Theorem \ref{thm_compare_P} and the division of $\{k_{l},l\geq 1\}$, it is straightforward to guarantee  (\ref{thm_compare}). \textbf{Q.E.D.}
	\end{proof}
	
	{\color{black}
		Under the boundedness of $P_{k,i}$ provided in Theorem \ref{thm_consistent}, we can obtain Theorem \ref{thm_noisefree} which depicts the convergence of the proposed CDKF.
		\begin{theorem}\label{thm_noisefree}
			Under Assumptions \ref{ass_noise}--\ref{ass_topology}, if $\{A_{k}\}_{k=0}^{\infty}$ belongs to a nonsingular compact set, then there is
			\begin{equation}\label{zero_covergence}
			\lim\limits_{k\rightarrow +\infty}E\{\hat x_{k,i}-x_{k}\}=0, \quad   \forall i\in \mathcal{V}.
			\end{equation}
		\end{theorem}
		\begin{proof}
			Under the conditions of this theorem, the conclusion of Theorem \ref{thm_consistent} holds. Thus, $P_{k,i}$ is uniformly upper bounded.
			Then according to Assumptions \ref{ass_noise} and \ref{ass_A_bound}, $\bar P_{k,i}$ is uniformly upper bounded and lower bounded.
			Thus we can define the following Lyapunov function
			\begin{equation*}
			V_{k,i}(x)=x^T\bar P_{k,i}^{-1}x.
			\end{equation*}
			From the fact (\expandafter{\romannumeral3}) of Lemma 1 in \cite{Battistelli2014Kullback} and the invertibility of $A_{k}$, there is
			\begin{equation}\label{iteration_V1}
			\begin{split}
			&V_{k+1,i}(E\{\bar e_{k+1,i}\})\\
			=&E\{\bar e_{k+1,i}\}^T\bar P_{k+1,i}^{-1}E\{\bar e_{k+1,i}\}\\
			\leq& \check{\beta}E\{\bar e_{k+1,i}\}^TA_{k}^{-T}P_{k,i}^{-1}A_{k}^{-1}E\{\bar e_{k+1,i}\},
			\end{split}
			\end{equation}
			where $0<\check{\beta}<1$. 
			
			{\color{black}
				Due to $\bar e_{k+1,i}=A_{k}e_{k,i}-w_{k}$, $E\{w_{k}\}=0$ and (\ref{iteration_V1}), one can obtain}
			\begin{equation}\label{V_ineq}
			V_{k+1,i}(E\{\bar e_{k+1,i}\})\leq \check{\beta}E\{e_{k,i}\}^TP_{k,i}^{-1}E\{e_{k,i}\}.
			\end{equation}
			Notice that $P_{k,i}=(\sum_{j\in \mathcal{N}_{i}}w_{k,i,j}\tilde P_{k,j}^{-1})^{-1}$,
			%
			there is
			\begin{equation}\label{x_iteration2}
			x_{k}=P_{k,i}\sum_{j\in \mathcal{N}_{i}}w_{k,i,j}\tilde P_{k,j}^{-1} x_{k}.
			\end{equation}
			Hence, according to CDKF and  (\ref{x_iteration2}), the estimation error $e_{k,i}$ satisfies
			\begin{equation}\label{ek}
			\begin{split}
			&E\{e_{k,i}\}\\
			=&E\{\hat x_{k,i}-x_{k}\}\\
			=&P_{k,i}\sum_{j\in \mathcal{N}_{i}}w_{k,i,j}\tilde P_{k,j}^{-1}E\{\phi_{k,j}-x_{k}\}\\
			=&P_{k,i}\sum_{j\in \mathcal{N}_{i}}w_{k,i,j}\tilde P_{k,j}^{-1}(I-K_{k,j}H_{k,j})E\{\bar e_{k,j}\}.
			\end{split}
			\end{equation}
			Since $\tilde P_{k,i}=(I-K_{k,i}H_{k,i})\bar P_{k,i}$, one can get
			\begin{equation}\label{P_imd}
			\begin{split}
			\tilde P_{k,i}^{-1}(I-K_{k,i}H_{k,i})=\bar P_{k,i}^{-1}.
			\end{split}
			\end{equation}
			Substituting  (\ref{P_imd}) into  (\ref{ek}), we have
			\begin{equation}\label{ek2}
			\begin{split}
			E\{e_{k,i}\}=P_{k,i}\sum_{j\in \mathcal{N}_{i}}w_{k,i,j}\bar P_{k,j}^{-1}E\{\bar e_{k,j}\}.
			\end{split}
			\end{equation}
			Since $P_{k,i}=(\sum_{j\in \mathcal{N}_{i}}w_{k,i,j}\tilde P_{k,j}^{-1})^{-1}$, there is
			\begin{equation}\label{P_inv}
			\begin{split}
			P_{k,i}&=(\sum_{j\in \mathcal{N}_{i}}w_{k,i,j}(\bar P_{k,j}^{-1}+H_{k,j}^TR_{k,j}^{-1}H_{k,j}))^{-1}\\
			&\leq (\sum_{j\in \mathcal{N}_{i}}w_{k,i,j}\bar P_{k,j}^{-1})^{-1}.
			\end{split}
			\end{equation}
			Applying  (\ref{ek2}),   (\ref{P_inv}) and Lemma 2 in \cite{Battistelli2014Kullback} to the right hand of  (\ref{V_ineq}), one can obtain that
			%
			\begin{equation}\label{V_ineq_final3}
			\begin{split}
			&V_{k+1,i}(E\{\bar e_{k+1,i}\})\\
			\leq&\check{\beta}\sum_{j\in \mathcal{N}_{i}}w_{k,i,j}E\{\bar e_{k,j}\}^T\bar P_{k,j}^{-1}E\{\bar e_{k,j}\}\\
			\leq&\check{\beta}\sum_{j\in \mathcal{N}_{i}}w_{k,i,j}V_{k,j}(E\{\bar e_{k,j}\}).
			\end{split}
			\end{equation}
			Denote $\mathcal{A}_{k}=[w_{k,i,j}],i,j=1,2,\cdots,N$. Summing up  (\ref{V_ineq_final3}) for $i=1,2,\cdots,N$, then there is
			\begin{equation}\label{V_ineq_final4}
			\begin{split}
			V_{k+1}(E\{\bar e_{k+1}\})
			&{\color{black}\leq \check{\beta}\mathcal{A}_{k}V_{k}(E\{\bar e_{k}\}), \quad 0<\check{\beta}<1,}
			\end{split}
			\end{equation}
			where  
			\begin{align*}
			V_{k}(E\{\bar e_{k}\})=col\{V_{k,1}(E\{\bar e_{k,1}\}),\cdots,V_{k,N}(E\{\bar e_{k,N}\})\}.
			\end{align*}
			According to Lemma \ref{lemma_1}, $\mathcal{A}_{k}$ is a row stochastic matrix at each moment, thus the spectral radius of $\mathcal{A}_{k}$ is always 1. 
			Due to $0<\check{\beta}<1$, $\lim\limits_{k\rightarrow +\infty}E\{\bar e_{k+1,i}\}= 0$. Under the equation $E\{\bar e_{k+1,i}\}=A_{k}E\{e_{k,i}\}$ and the assumption that $\{A_{k}\}_{k=0}^{\infty}$ belongs to a nonsingular compact set, 
			the conclusion of this theorem holds. \textbf{Q.E.D.}
			%
			%
		\end{proof}
	}
	\begin{remark}
		{\color{black}The reason for the invertibility of $A_{k}$ in Theorem  \ref{thm_noisefree} is} that the proof using Lyapunov method to guarantee the convergence needs the $\bar P_{k,i}^{-1}$ (see  (\ref{iteration_V1})), whose iteration requires the invertibility of $A_{k}$. While, in the proof of the boundedness of covariance matrix in Theorem \ref{thm_consistent}, we used another method to relax the invertibility of $A_{k}$.
	\end{remark}

	\section{Simulation Studies}
	To demonstrate the aforementioned theoretical results, two simulation examples in this section will be studied. In the first example, both the consistency and the boundedness of covariance matrix  will be illustrated. The performance of adaptive CI weights will be compared with that of constant CI weights as well as the algorithm in Table \ref{ODKF1}. In the second example, the proposed algorithm is compared with some other algorithms.
	
	\subsection{Performance evaluation}
	Consider the following second-order time-varying stochastic system with four sensors in the network
	\begin{equation*}
	\begin{cases}
	x_{k+1}=\left(
	\begin{array}{cc}
	1.1 & 0.05 \\
	1.1 & 0.1sin(\frac{k\pi}{6}) \\
	\end{array}
	\right)
	x_{k}+\omega_{k},\\
	y_{k,i}=H_{k,i}x_{k}+v_{k,i},i=1,2,3,4,
	\end{cases}
	\end{equation*}
	where the observation matrices of the sensors are
	\begin{align*}
	\begin{cases}
	H_{k,1}=\left(
	\begin{array}{cc}
	1+sin(\frac{k\pi}{12}) & 0
	\end{array}
	\right),\\
	H_{k,2}=\left(
	\begin{array}{cc}
	0 & 0 \\
	\end{array}
	\right),\\
	H_{k,3}=\left(
	\begin{array}{cc}
	-1 & 1+cos(\frac{k\pi}{12}) \\
	\end{array}
	\right),\\
	H_{k,4}=\left(
	\begin{array}{cc}
	0 & 0 \\
	\end{array}
	\right).
	\end{cases}
	\end{align*}
	Here, it is assumed that the process noise covariance matrix $Q_{k}=diag\{0.5,0.7\}$, and the whole measurement noise covariance matrix $R_{k}=diag\{0.5,0.6,0.4,0.3\}$. The initial value of the state is generated by a Gaussian process with zero mean and covariance matrix $I_2$, and the initial estimation settings are $\hat x_{i,0}=0$ and $P_{i,0}=I_2$, $\forall i=1,2,3,4$. The sensor network's communication topology, assumed as directed and strongly connected, is illustrated in Fig. \ref{topology}. The weighted adjacent matrix $\mathcal{A}=[a_{i,j}]$ is designed as
	$a_{i,j}=\frac{1}{|\mathcal{N}_{i}|},j\in\mathcal{N}_{i},i,j=1,2,3,4$. 
	We conduct the numerical simulation through Monte Carlo experiment, in which $500$ Monte Carlo trials  are performed. The mean square error of a whole network is defined as
	\begin{equation*}
	MSE_{k}=\frac{1}{N}\sum_{i\in\mathcal{V}}\frac{1}{500}\sum_{j=1}^{500}(\hat x_{k,i}^j-x_{k}^j)^T(\hat x_{k,i}^j-x_{k}^j),
	\end{equation*}
	where $\hat x_{k,i}^j$ is the state estimation of the $j$th trail of Sensor $i$ at the $k$th moment.

	To show the ability of CDKF in coping with the singularity of system matrices, the determinants of time-varying  system matrices are plotted in Fig. \ref{A_singular}. The tracking graph for the system states is shown in Fig. \ref{tracking}. From Fig. \ref{A_singular} and Fig. \ref{tracking}, it can be seen that under the case that system matrices are singular, the proposed algorithm CDKF has effective tracking performance for each state element. Fig. \ref{tracking} also shows the unbiasedness of estimations conforming with Theorem  \ref{thm_distri}.
	The consistency of CDKF is depicted in Fig. \ref{consictence}, where $MSE_{k}$ is compared with $tr(\sum_{i=1}^{4}P_{k,i})$. From this figure,  it is found that the CDKF keeps stable in the given period and the estimation error can be evaluated in real time. Besides, Fig. \ref{consictence} also shows the comparison between the algorithm in Table \ref{ODKF1} and the proposed CDKF in Table \ref{ODKF2}. We can see that although the CDKF is sub-optimal, the estimation performance of CDKF is very close  to the networked Kalman filter with optimal gain parameter which utilizes the global information. 
	To test the effectiveness of the SDP optimization algorithm for adaptive weights in the CI strategy, we compare the algorithms with adaptive CI  weights and constant CI weights. The comparison results are shown in Fig. \ref{comparisionwk1} and Fig. \ref{comparisionwk}. Fig. \ref{comparisionwk1} shows that the results on parameter matrix $P$ in the two algorithms conform with the aforementioned theoretical analysis in Theorem \ref{thm_compare_P}.
	Fig. \ref{comparisionwk} implies that the optimization algorithm for adaptive weights is quite effective in improving the estimation performance.

	The above results reveal that the proposed CDKF is an effective and flexible distributed state estimation algorithm.
	
	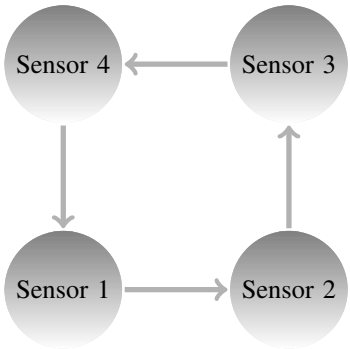
\begin{figure}[htp]
		\centering
		\begin{tikzpicture}[scale=1, transform shape,line width=2pt]
		\tikzstyle{every node} = [circle,shade, fill=gray!30]
		\node (a) at (0, 0) {Sensor 1};
		\node (b) at +(0: 1.5*2) {Sensor 2};
		\node (c) at +(45: 2.1213*2) {Sensor 3};
		\node (d) at +(90: 1.5*2) {Sensor 4};
		\foreach \from/\to in {a/b, b/c, c/d, d/a}
		\draw [black!30,->] (\from) -- (\to) ;
		\end{tikzpicture}
		\caption{The topology of the sensor network}\label{topology}
	\end{figure}

	
	\begin{figure}
		\centering
		\includegraphics[width=0.5\textwidth]{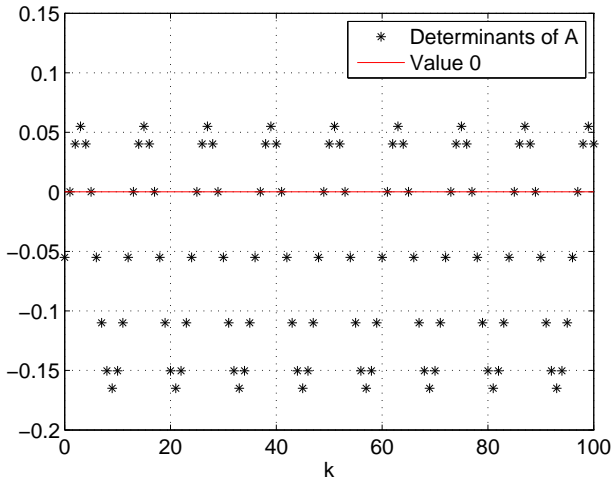}
		\caption{Determinants of system matrix A: singular points}
		\label{A_singular}
	\end{figure}

	\begin{figure}
		\centering
		\includegraphics[width=0.5\textwidth]{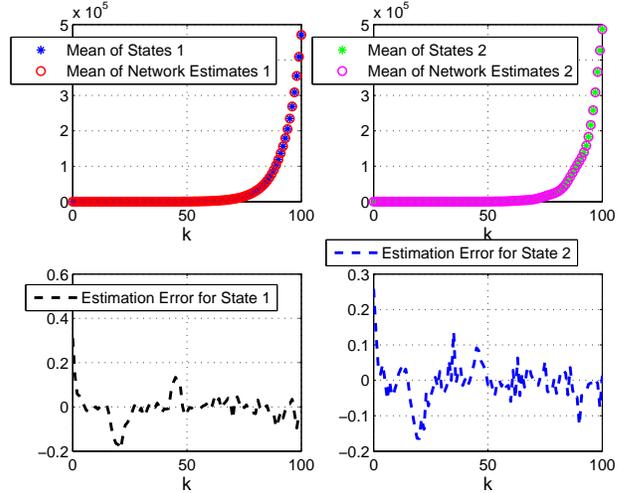}
		\caption{Network tracking for each state}
		\label{tracking}
	\end{figure}
	\begin{figure}
		\centering
		\includegraphics[width=0.5\textwidth]{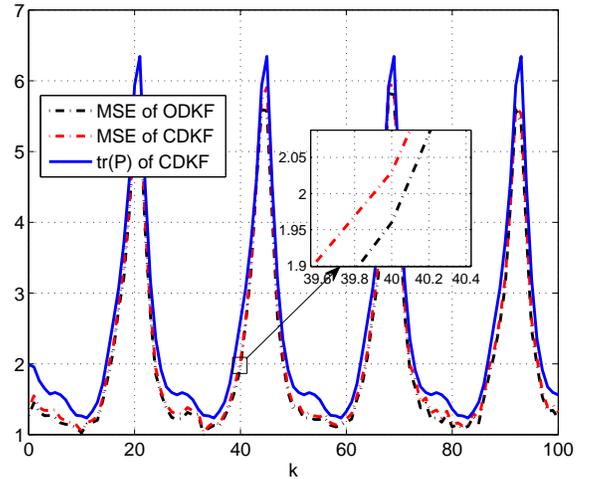}
		\caption{Network Estimation of the algorithm in Table \ref{ODKF1} and CDKF in Table \ref{ODKF2}}\label{consictence}
	\end{figure}
	\begin{figure}
		\centering
		\includegraphics[width=0.5\textwidth]{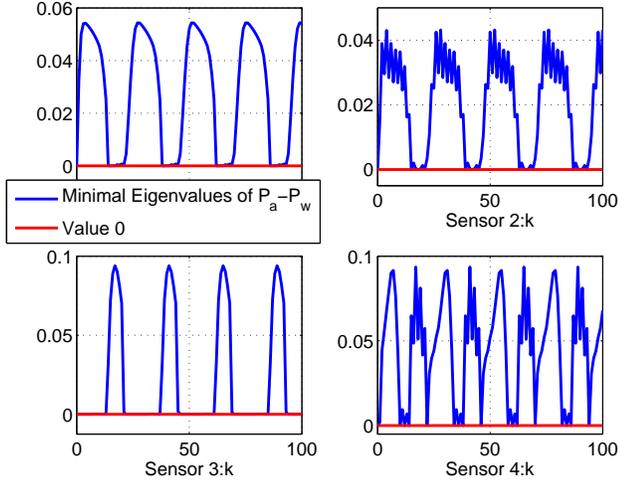}
		\caption{Minimal eigenvalues of $P_a-P_w$ for each sensor}
		\label{comparisionwk1}
	\end{figure}
	\begin{figure}
		\centering
		\includegraphics[width=0.5\textwidth]{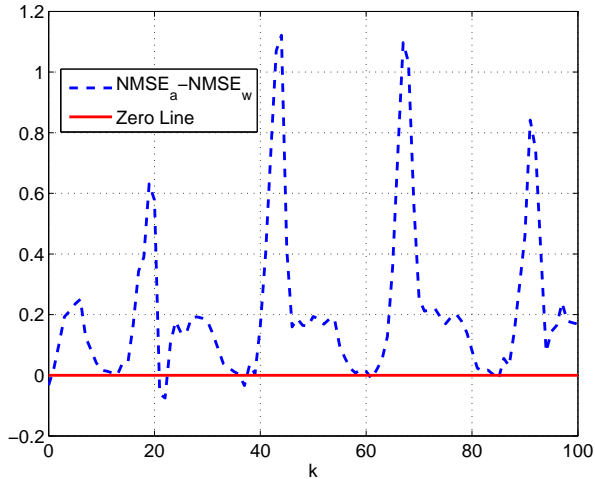}
		\caption{Network MSE (NMSE) of two kinds of weights}\label{comparisionwk}
	\end{figure}
	\subsection{Comparisons with other algorithms}
	In this subsection, numerical simulations are carried out to compare the proposed CDKF with some other algorithms including CKF, Collaborative Scalar-gain Estimator (CSGF) \cite{Khan2014Collaborative} and Distributed State
	Estimation with Consensus on the Posteriors (DSEA-CP) \cite{Battistelli2014Kullback}, which is a information form of distributed Kalman filter.
	The algorithm CKF is the optimal centralized algorithm.
	Here, we consider the simulation example studied in \cite{Khan2014Collaborative} on CSGF. This example is based on time-invariant systems, i.e., for system (\ref{system_all}), $A_{k}=A$, $Q_{k}=Q$, $H_{k,i}=H_{i}$ and $R_{k,i}=R_{i},\forall i\in \mathcal{V},\forall k=0,1,\ldots$.
	The topology of the sensor network consisting of 20 sensors, assumed as undirected and connected,  is illustrated in Fig. \ref{topology2}.
	The weighted adjacent matrix $\mathcal{A}=[a_{i,j}]$ is designed as
	$a_{i,j}=\frac{1}{|\mathcal{N}_{i}|},j\in\mathcal{N}_{i},i,j=1,\ldots,N$. The system matrices are assumed to be $Q=diag\{1,1\},\text{ }R_i=1,i\in\mathcal{V}$ and $
	A=\left(
	\begin{array}{cc}
	1 & 0.05 \\
	0 & 1 \\
	\end{array}
	\right)
	$.
	The observation matrices of these sensors are uniformly randomly selected from $\{(1,1),(0,0),(0,0),(1,0)\}$.
	The initial value of the state is generated by a Gaussian process with zero mean and covariance matrix $I_2$, and the initial estimation settings are $\hat x_{i,0}=0$ and $P_{i,0}=I_2$, $\forall i\in\mathcal{V}$.
	We conduct the numerical simulation through Monte Carlo experiment, in which
	$500$ Monte Carlo trials for CKF, CDKF, CSGF and DSEA-CP are performed, respectively.
	The comparison of estimation error dynamic is carried out for the four algorithms, and the result can be seen in Fig. \ref{comparision_case1}.
	From this figure, we see that the four algorithms are all stable and the estimation performance of CDKF is better than CSGF as well as DSEA-CP. The reason that  CKF, CDKF and DSEA-CP have better estimation performance than CSGF is the time-varying gain matrices in the measurement updates. Additionally, the estimation performance of CDKF is nearer to CKF.

	\begin{figure}
		\centering
		\includegraphics[width=0.5\textwidth]{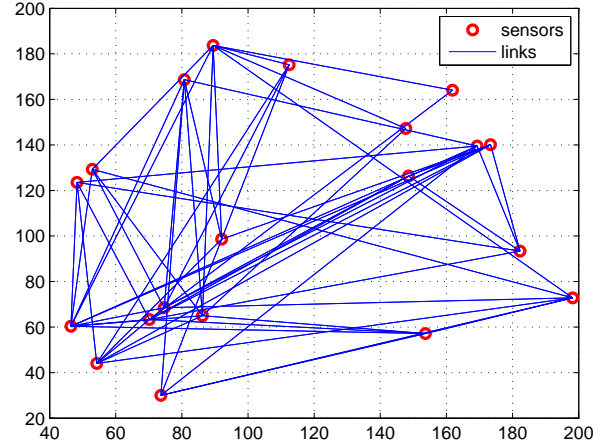}
		\caption{The topology of the sensor network with 20 sensors}
		\label{topology2}
	\end{figure}
	\begin{figure}
		\centering
		\includegraphics[width=0.5\textwidth]{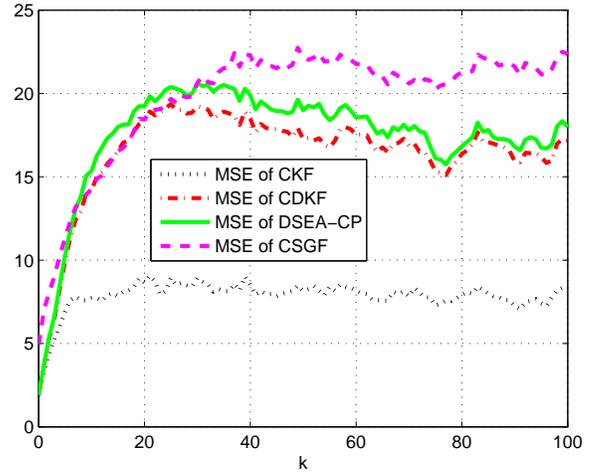}
		\caption{The performance comparison of filters}
		\label{comparision_case1}
	\end{figure}

	%

	
	\section{Conclusion}
	This paper has investigated the distributed state estimation problem for a class of discrete time-varying systems.
	Since the networked Kalman filter with optimal gain parameter needs the global information which is usually forbidden in large networks,  a sub-optimal  distributed Kalman filter based on CI fusion method was proposed.
	The consistency of the algorithm can provide an effective method to evaluate the estimation error in real time. In order to improve the estimation performance at each moment, the design of adaptive CI weights was casted through an optimization problem, which can be solved with a convex SDP optimization method.
	Additionally, it was
	proven that the adaptive CI weights can give rise to a lower error covariance bound than the constant CI weights.
	{\color{black}For the proposed algorithm,
		the boundedness of covariance matrix and the convergence  have been analyzed based on the global observability condition and the strong connectivity of the directed network topology, which are  general requirements for distributed state estimations.}
	Additionally, the proposed algorithm has loosen the  nonsingularity of system matrix in the main distributed filter designs under the condition of global observability.
	{\color{black}The simulation examples have shown the effectiveness of the proposed algorithm in the considered scenarios.}

	\bibliographystyle{plain}        

	\bibliography{references_filtering}           
	
	
	
	\appendix
	\section{Proof of Theorem \ref{thm_distri}}\label{pf_dis}
	According to the CDKF of Table \ref{ODKF2}, the estimation  error is
	\begin{align}\label{e_k0}
	e_{k,i}&=\hat x_{k,i}-x_{k}=P_{k,i}\sum_{j\in \mathcal{N}_{i}}w_{k,i,j}\tilde P_{k,j}^{-1}(\phi_{k,j}-x_{k})\nonumber\\
	&=P_{k,i}\sum_{j\in \mathcal{N}_{i}}w_{k,i,j}\tilde P_{k,j}^{-1}\tilde e_{k,j},
	\end{align}
	where $\tilde e_{k,i}$ is the estimation error in the measurement update, which can be derived through
	\begin{align}\label{e_k1}
	\tilde e_{k,i}&=\phi_{k,i}-x_{k}=\bar x_{k,i}+K_{k,i}(y_{k,i}-H_{k,i}\bar x_{k,i})-x_{k}\nonumber\\
	&=(I-K_{k,i}H_{k,i})\bar e_{k,i}+K_{k,i}v_{k,i}.
	\end{align}
	The prediction error $\bar e_{k,i}$ follows from
	\begin{equation}\label{e_bar1}
	\begin{split}
	\bar e_{k,i}&=\bar x_{k,i}-x_{k}=A_{k-1}e_{k-1,i}-\omega_{k-1}.
	\end{split}
	\end{equation}	
	Thus, from  (\ref{e_k0})-(\ref{e_bar1}), there is
	\begin{align}\label{Eq_guassian}
	e_{k,i}&=P_{k,i}\sum_{j\in \mathcal{N}_{i}}w_{k,i,j}\tilde P_{k,j}^{-1}\tilde e_{k,j}\nonumber\\
	&=P_{k,i}\sum_{j\in \mathcal{N}_{i}}w_{k,i,j}\tilde P_{k,j}^{-1}\big((I- K_{k,j}H_{k,j})\bar e_{k,j}+K_{k,j}v_{k,j}\big)\nonumber\\
	&=P_{k,i}\sum_{j\in \mathcal{N}_{i}}w_{k,i,j}\tilde P_{k,j}^{-1}(I- K_{k,j}H_{k,j})A_{k-1}e_{k-1,j}\nonumber\\
	&\quad-P_{k,i}\sum_{j\in \mathcal{N}_{i}}w_{k,i,j}\tilde P_{k,j}^{-1}(I- K_{k,j}H_{k,j})\omega_{k-1}\nonumber\\
	&\quad+P_{k,i}\sum_{j\in \mathcal{N}_{i}}w_{k,i,j}\tilde P_{k,j}^{-1}K_{k,j}v_{k,j}.
	\end{align}		
	Thanks to the  Gaussian property acting on  (\ref{Eq_guassian}), and the Gaussianity of $e_{0,i}$, $\omega_{k-1}$ and $v_{k,i}$,  $e_{k,i}$ is also Gaussian.		
	Due to $E\{\omega_{k-1}\}=0$ and $E\{v_{k,i}\}=0,$	from  (\ref{Eq_guassian}) one can obtain that
	\begin{equation*}
	\begin{split}
	&E\{e_{k,i}\}\\
	=&P_{k,i}\sum_{j\in \mathcal{N}_{i}}w_{k,i,j}\tilde P_{k,j}^{-1}(I- K_{k,j}H_{k,j})A_{k-1}E\{e_{k-1,j}\}.
	\end{split}
	\end{equation*}	
	Under the initial condition of this algorithm, there is $E\{e_{0,i}\}=0,\forall i\in\mathcal{V}$. Therefore, it is straightforward to rectify  (\ref{lem_unbias}). \textbf{Q.E.D.}
	
	{\color{black}
		\section{Proof of Lemma \ref{lemma_1}}\label{pf_lemma1} 
		According to Lemma \ref{Schur_Complement},  (\ref{eq_constraint1}) is equivalent to 
		\begin{align*}
		\Delta_{k,i}>0, M_{k,i}-\Delta_{k,i}^{-1}>0.
		\end{align*}
		Then we have $M_{k,i}>\Delta_{k,i}^{-1}$. Hence, $tr(M_{k,i})>tr(\Delta_{k,i}^{-1})$. Under the condition, the parameter set $\{w_{k,i,j},m_{k,i_l}\}$ that minimizes $tr(M_{k,i})$ will lead to the minimization of $tr(\Delta_{k,i}^{-1})$, and vice versa. \textbf{Q.E.D.}
		
		\section{Proof of Lemma \ref{lemma_2}}\label{pf_lemma2} 
		First, considering the objective function (\ref{objective_weight}), we set a vector $z$ with the form in Table \ref{SDP} to rewrite the objective function with the SDP algorithm form. After getting optimized $z^*$, we can obtain $\{w_{k,i,j}\}$ through the equation in 2) of Table \ref{SDP}.
		Then, the LMI (\ref{eq_constraint1}) is written as the third inequality constraint in Table \ref{SDP}, which consists of three types of $F_{j}$ listed in the table. Finally, the first and the second inequality constraints in Table \ref{SDP} correspond to the constraints $\sum_{j\in \mathcal{N}_{i}} w_{k,i,j}=1$ and $0\leq  w_{k,i,j}\leq 1,j\in \mathcal{N}_{i}$, respectively. Therefore, we can obtain the general form of SDP algorithm in Table \ref{SDP}. \textbf{Q.E.D.}

		\section{Proof of Theorem \ref{thm_compare_P}}  \label{pf_compare}           
		Here we utilize the inductive method to give the proof of this theorem.
		Firstly, at the $(k-1)$th moment, it is supposed that $P_{k-1,i|w}\leq P_{k-1,i|a},\forall i\in \mathcal{V}.$
		Then due to the prediction equation $\bar P_{k,i}=A_{k-1}P_{k-1,i}A_{k-1}^T+Q_{k-1},$
		there is
		\begin{align}\label{compare_1}
		\bar P_{k,i|w}\leq\bar P_{k,i|a}.
		\end{align}
		Exploiting the matrix inverse formula on the measurement update equation of CDKF, the following iteration holds
		\begin{align}
		&\tilde P_{k,i}^{-1}=\bar P_{k,i}^{-1}+H_{k,i}^TR_{k,i}^{-1}H_{k,i}.\label{equ_short}
		\end{align}
		Then from  (\ref{compare_1}),  we have $\tilde P_{k,i|w}^{-1}=\bar P_{k,i|w}^{-1}+H_{k,i}^T R_{k,i}^{-1}H_{k,i}\\
		\geq \bar P_{k,i|a}^{-1}+H_{k,i}^TR_{k,i}^{-1}H_{k,i}=\tilde P_{k,i|a}^{-1}.$
		Due to the equation in the local fusion process $P_{k,i}=(\sum_{j\in \mathcal{N}_{i}}w_{k,i,j}\tilde P_{k,j}^{-1})^{-1},$
		and the optimization condition for adaptive CI weights in (\ref{eq_constraint1}),
		there is
		\begin{equation*}
		\begin{split}
		P_{k,i|w}^{-1}&=\sum_{j\in \mathcal{N}_{i}}w_{k,i,j}\tilde P_{k,j|w}^{-1},\geq\sum_{j\in \mathcal{N}_{i}}a_{i,j}\tilde P_{k,j|w}^{-1}\\
		&\geq\sum_{j\in \mathcal{N}_{i}}a_{i,j}\tilde P_{k,j|a}^{-1}=P_{k,i|a}^{-1}.
		\end{split}
		\end{equation*}
		Therefore, $P_{k,i|w}\leq P_{k,i|a}.$
		The proof is finished with the initial conditions satisfying $P_{0,i|w}=P_{0,i|a}.$ \textbf{Q.E.D.}
	}
\end{document}